\documentclass[a4paper,USenglish,10pt]{article}
\usepackage[utf8]{inputenc}

\usepackage{amsthm,amssymb,amsfonts,amsmath}
\usepackage{amssymb,amsfonts,amsmath}
\newtheorem{theorem}{Theorem}
\newtheorem{lemma}[theorem]{Lemma}

\newtheorem{proposition}[theorem]{Proposition}

\theoremstyle{definition}

\usepackage[hidelinks]{hyperref}
\usepackage[blocks]{authblk}

\usepackage{microtype}
\usepackage[noend, linesnumbered]{algorithm2e}
\usepackage{color}
\usepackage{cite}
\usepackage{pifont}
\usepackage{enumerate}

\usepackage{url}

\makeatletter
\newcommand{\ourvec}[1]{\vec{#1}\@ifnextchar{^}{\,}{}}
\makeatother

\newcommand{\ourbound}{42.11}

\makeatletter
\g@addto@macro\bfseries{\boldmath}
\makeatother

\usepackage{graphicx}
\graphicspath{{figures/}}

\usepackage[mathlines]{lineno}
\newcommand*\patchAmsMathEnvironmentForLineno[1]{%
\expandafter\let\csname old#1\expandafter\endcsname\csname #1\endcsname
\expandafter\let\csname oldend#1\expandafter\endcsname\csname end#1\endcsname
\renewenvironment{#1}%
     {\linenomath\csname old#1\endcsname}%
     {\csname oldend#1\endcsname\endlinenomath}}%
\newcommand*\patchBothAmsMathEnvironmentsForLineno[1]{%
  \patchAmsMathEnvironmentForLineno{#1}%
  \patchAmsMathEnvironmentForLineno{#1*}}%
\AtBeginDocument{%
\patchBothAmsMathEnvironmentsForLineno{equation}%
\patchBothAmsMathEnvironmentsForLineno{align}%
\patchBothAmsMathEnvironmentsForLineno{flalign}%
\patchBothAmsMathEnvironmentsForLineno{alignat}%
\patchBothAmsMathEnvironmentsForLineno{gather}%
\patchBothAmsMathEnvironmentsForLineno{multline}%
}
\date{}

\newcommand{\ceil}[1]{\ensuremath{\left \lceil #1 \right \rceil}}

\newcommand\blfootnote[1]{%
  \begingroup
  \renewcommand\thefootnote{}\footnote{#1}%
  \addtocounter{footnote}{-1}%
  \endgroup
}

\usepackage{color}

\title{A new lower bound on the maximum number of plane graphs using production matrices}

\author[1]{Clemens Huemer}
\author[2]{Alexander Pilz}
\author[1]{Rodrigo I. Silveira}
\affil[1]{Departament de Matem\`atiques, Universitat Polit\`ecnica de Catalunya, Barcelona, Spain\\
  \{\texttt{clemens.huemer,rodrigo.silveira}\}\texttt{@upc.edu}.}
\affil[2]{Institute of Software Technology, Graz University of Technology, Austria\\
  \texttt{apilz@ist.tugraz.at}.}

\newcommand{\fig}[1]{\figurename~\ref{#1}}

\makeatletter
\g@addto@macro\bfseries{\boldmath}
\makeatother

\advance\hoffset-7mm
\begin{document}

\maketitle

\begin{abstract}
We use the concept of production matrices to show that there exist sets of $n$ points in the plane that admit $\Omega(\ourbound^n)$ crossing-free geometric graphs.
This improves the previously best known bound of $\Omega(41.18^n)$ by Aichholzer et al.~(2007).
\end{abstract}

\section{Introduction}
A \emph{geometric graph} on a set $S$ of $n$ labeled points in the Euclidean plane is a graph with vertex set $S$ where each edge is represented by a straight line segment between the corresponding points.\blfootnote{\begin{minipage}[l]{0.2\textwidth} \vspace{-8pt}\includegraphics[trim=10cm 6cm 10cm 5cm,clip,scale=0.15]{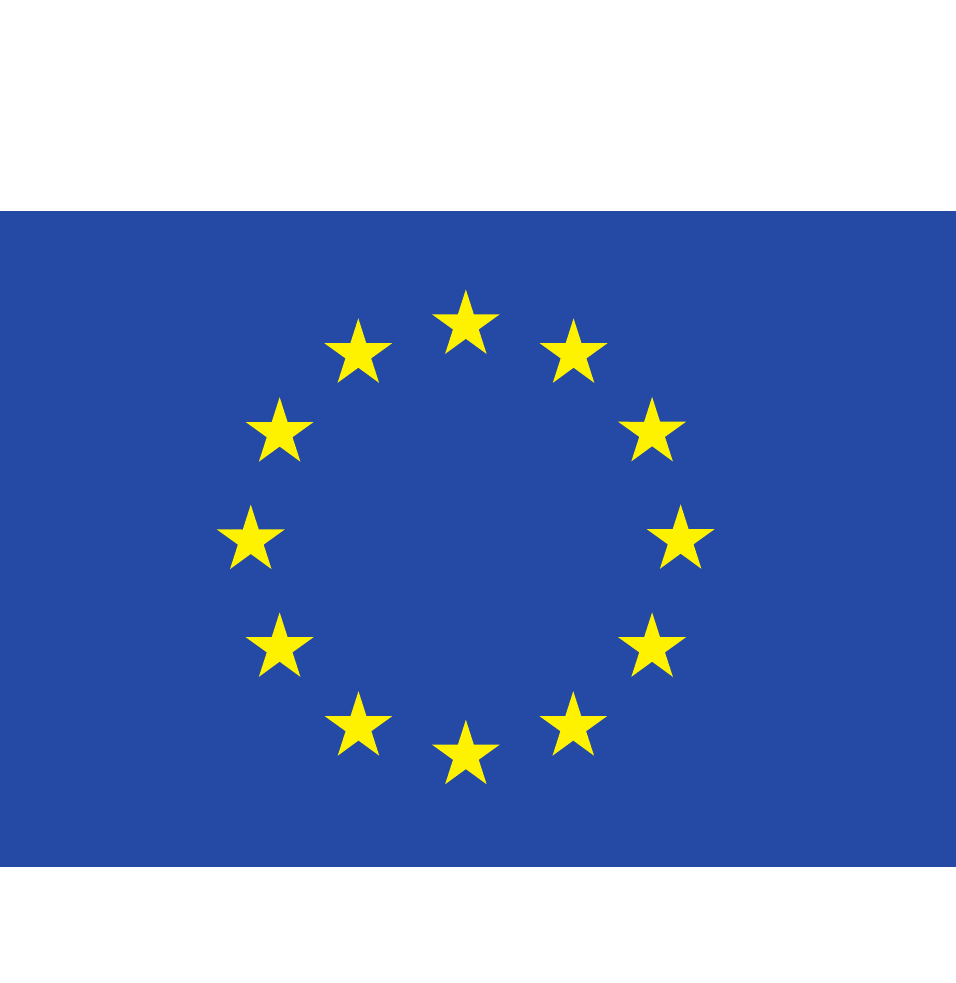} \end{minipage}  \hspace{-1.3cm} \begin{minipage}[l][1cm]{0.82\textwidth}
      This project has received funding from the European Union's Horizon 2020 research and innovation programme under the Marie Sk\l{}odowska-Curie grant agreement No 734922.
     \end{minipage}}
\blfootnote{C.~H.\ and R.~S.\ were also supported by projects MINECO MTM2015-63791-R and Gen.\ Cat.\ 2017SGR1336 and 2017SGR1640, respectively. R.~S.\ was further supported by MINECO through the Ram{\'o}n y Cajal program. A.~P.\ is supported by a Schr\"odinger fellowship of the Austrian Science Fund (FWF): J-3847-N35.
This work was done while A.~P.\ was at the Department of Computer Science of ETH Z\"urich.}
In this work, we are interested in the number of \emph{crossing-free} geometric graphs on a set of $n$~points, i.e., geometric graphs in which all segments are interior-disjoint, also referred to as \emph{plane graphs}.
It is easy to see that, for any $n$ points, this number is at least exponential in~$n$.
In 1982, Ajtai et al.~\cite{ajtai} showed that the upper bound on this number is also exponential.
Currently, it is known that any set of $n$ points admits not more than $O(187.53^n)$ crossing-free graphs~\cite{sharir_sheffer_charging}.
(Bounds on the number of graphs are usually stated for point sets in general position, i.e., without any three points on a line.
Further, as the points representing vertices are distinct, geometric graphs are considered labeled in the context of counting.)
While the number of crossing-free graphs is minimized if the point set is in convex position~\cite{number_plane_geometric}, not much is known about sets maximizing this number.
The best known example by now is the so-called \emph{double-zig-zag chain}~\cite{number_plane_geometric}, with $\Omega(41.18^n)$ crossing-free graphs.
As usual, such lower-bound constructions rely on describing a family of point sets with convenient structural properties.
In this paper, we improve this bound by showing that another well-known family of point sets, a generalization of the double-zig-zag chain, admits $\Omega(\ourbound^n)$ crossing-free graphs.
This generalization has also been used for similar bounds on triangulations~\cite{bounds_multiplicity} and, recently, on crossing-free perfect matchings~\cite{asinowsi_matchings}, but the number of general crossing-free graphs on this configuration was not known.
The method that allows us to analyze these point sets is the use of \emph{production matrices}, a technique that we consider interesting on its own.

This method works by implicitly arranging the graphs in a \emph{generating tree}, describing a rule to produce a graph from one on fewer points.
We consider a partition of the set of graphs on $i \leq n$ points into $n$ parts according to their degree at an arbitrarily defined \emph{root vertex}, and represent the cardinality of each part in a vector $\ourvec{v}^{(i)}$.
The first element of  $\ourvec{v}^{(i)}$ is the number of graphs with the root vertex having degree~0, the second one that of graphs with root vertex with degree~1, and so on.
We then devise how to generate graphs on $i+c$ points (for some small positive number~$c$) with a new root vertex, from the graphs counted in $\ourvec{v}^{(i)}$, and again give the cardinalities of their parts in a vector $\ourvec{v}^{(i+c)}$.
In the production matrix approach, the relation between $\ourvec{v}^{(i)}$ and $\ourvec{v}^{(i+c)}$ is encoded in an $n \times n$ \emph{production matrix} $A \in \mathbb{Q}_0^{n \times n}$ such that $\ourvec{v}^{(i+c)} = A \ourvec{v}^{(i)}$.
In this way, we obtain the number of graphs on $n$ vertices in $\ourvec{v}^{(n)} = A^j\ourvec{v}^{(n_0)}$ from the graphs on a constant number $n_0$ of vertices, with $j = (n - n_0)/c$.
In this paper we show how that can be done for crossing-free geometric graphs on point sets with a particular structure.

We focus on obtaining an asymptotic lower bound on the number of crossing-free graphs. 
To that end, we obtain the corresponding production matrix~$A$, and apply the Perron--Frobenius theorem to obtain a lower bound on the elements of $A^j$ when $j$ tends to infinity, by approximating the largest eigenvalue of the matrix.
This gives us a lower bound on the number of crossing-free graphs on such a point set.

For points in convex position, generating trees have been described for triangulations~\cite{treeOfTriangulations}, spanning trees~\cite{treeOfTrees}, and very recently for a few other crossing-free graphs~\cite{production_matrices_geometric_graphs,guillermo}.
They are also the basis of the ECO method~\cite{eco_survey}.
The term \emph{production matrix} was introduced in~\cite{deutsch}, although the equivalent term \emph{AGT matrix}~\cite{MeVe} is sometimes used.
In a recent paper together with Seara~\cite{characteristic_polynomials}, we already studied characteristic polynomials of production matrices for various classes of geometric graphs, which can give rise to new relations between well-known combinatorial objects.
Asinowski and Rote~\cite{asinowsi_matchings} use similar matrices to bound the number of crossing-free perfect matchings of point sets; in fact, 
the classes of point sets they consider are the same as ours.

Indeed, the work by Asinowski and Rote~\cite{asinowsi_matchings} is closely related to ours, both in the classes of point sets considered, as well as in the counting methods.
In their paper, they use various methods for counting crossing-free perfect matchings on particular point sets; in Section~5 of their work, they obtain a sequence of infinite vectors whose elements are the number of matchings on $n$ vertices, partitioned by the number of unmatched points (i.e., points that still have to be matched).
They devise an infinite band matrix~$A$ to obtain this sequence of vectors.
Inspired by the Perron--Frobenius theorem (which we will use for our fixed-size matrices), they show that the growth rate for perfect matchings is equal to the column sum of~$A$ after stabilization.
This way, they also show that the bound is tight for matchings on that class of point sets.
While our production matrices can also be considered infinite, we will eventually consider constant-size matrices (and thus only obtain a subset of all possible crossing-free graphs), in order to obtain a lower bound using the Perron--Frobenius theorem.

For many classes of crossing-free graphs, it is known that their number is minimal when the points are in convex position~\cite{number_plane_geometric,convexity_pseudo_triangulations}.
A remarkable exception are triangulations, where it is conjectured that so-called \emph{double circles} are the minimizing point set class~\cite{bound_double_circle};
the current best bound~\cite{lower_bound_triangulations} is, however, far from the number of triangulations of the double circle.
Much less seems to be known about point sets that maximize the number of graphs.
See the online list by Sheffer~\cite{sheffer_webpage} for current bounds on these numbers for various graph classes.


\paragraph{Outline.} 
We begin in Section~\ref{sec:tutorial} by introducing the production matrix technique with an example, i.e., counting the number of plane graphs on points in convex position.
In the following section we define the family of point sets used to obtain our improved lower bound, the generalized double zig-zag chain.
In Section~\ref{sec:counting} we provide production matrices to count sub-graphs in one part of the point set;
the sub-graphs in the remaining part are counted by generalizing a previously-known technique.
In Section~\ref{sec:bounds_eigenvalue}, we argue that bounds on the Perron roots of the matrices give us a lower bound on the  number of crossing-free graphs, leading to our main result.

\section{Warm-up for points in convex position}\label{sec:tutorial}

The first stepping stone of our counting method will be a way to map a graph on $i+1$ vertices to one on $i$, and vice versa. 
Recall that  in this paper we only consider crossing-free graphs.
We will assume point sets are given as sequences $p_1, \dots, p_n$, where each $p_i$ ($1 \leq i \leq n$) is a vertex.

Consider any $(i+1)$-vertex graph~$G$ drawn on vertices $p_1$ to $p_{i+1}$.
We can associate $G$ to a graph $G'$ with $i$ vertices by replacing every edge $p_j p_{i+1}$ by the edge $p_j p_i$ for all $1 \leq j \leq i$ (discarding duplicates and loops).
The graph $G'$ that we obtain is called the \emph{parent} of~$G$.
In the other direction, we can select some edges incident to $p_i$ in $G'$ and replace them by edges incident to $p_{i+1}$ in a way that $G'$ is the parent of the new graph $\tilde G$, and such that $\tilde G$ is crossing-free.
We say that $G'$ \emph{produces} $\tilde G$, and that the edges incident to $p_{i+1}$ that are mapped to edges of $G'$ are \emph{inherited} from~$G'$.
The degree of $p_i$ in $G'$ determines how many graphs can be produced from it.

In our construction, we will refer to $p_i$ as the root vertex, and we will use a vector $\ourvec v^{(i)}$ to store the number of graphs with root vertex $p_i$ of degree $j$, for $0 \leq j \leq n$.
The relation between a graph on $n$ vertices and those that can be produced from it defines an implicit generating tree. 
For our purposes, we do not need the tree explicitly, but are only interested in counting how many graphs can be produced from another one.

To introduce the production matrix approach, we will show how to obtain the known lower bound for the number of graphs on $n$ points in convex position.
For point sets in convex position, the number of plane graphs was shown by Flajolet and Noy~\cite{FN} to be $\Theta((6+4\sqrt{2})^n/n^{3/2})$~\cite[Theorem~4]{FN}, which is in $\Omega(11.65^n)$.

The overall approach will be to derive a production matrix $C$ to count the number of graphs on $i+1$ points in convex position, based on those on $i$ points.
Once $C$ is found, we will apply the Perron-Frobenius theorem to show that the largest eigenvalue of $C$ gives a lower bound on the base of the expression for the asysmptotic number of graphs on $n$ points in convex position. As an illustration, matrix $C$ for $n=6$ has the following shape:
\begin{equation}\label{eq:matrix_convex}
\displaystyle C= \left(\begin{array}{rrrrrrrr}
1 & 1 & 1 & 1 & 1 & 1\\
1 & 2 & 2 & 2 & 2 & 2\\
0 & 1 & 2 & 2 & 2 & 2\\
0 & 0 & 1 & 2 & 2 & 2\\
0 & 0 & 0 & 1 & 2 & 2\\
0 & 0 & 0 & 0 & 1 & 2
\end{array}\right) \enspace.
\end{equation}


We have $n$ points $p_1, \dots, p_n$ in convex position, indexed from $1$ to $n$ in, say, clockwise order along the convex hull boundary.
Note that we can always add an edge of the convex hull to any graph without introducing any crossings.
Hence, we will focus on counting plane graphs without edges on the convex hull boundary, and at the end multiply their number by $2^n$, accounting for all possibilities of adding such edges (recall that we are considering labeled graphs).

Each graph on $i+1$ vertices is mapped to its unique parent graph on $i$ vertices that is obtained by identifying $p_{i+1}$ with $p_i$, and possibly deleting the edge $p_{i-1} p_i$ that is now on the convex hull boundary and stems from $p_{i-1} p_{i+1}$ (see \fig{fig:convex_chain}~(left)).
Note that apart from this edge, there is only one other possibility for the number of edges to be less in the parent graph than in the original one.
This happens when there are both the edges $p_k p_i$ and $p_k p_{i+1}$, which are mapped to the same edge (which is shown in \fig{fig:convex_chain}~(right)).
Observe that, since the graph is crossing-free, there is no other pair of such edges $p_{k'} p_i$ and $p_{k'} p_{i+1}$.

\begin{figure}[tb]
\centering
\includegraphics[page=3]{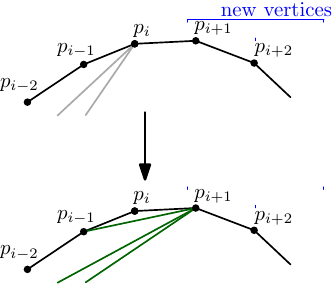}
\includegraphics[page=2]{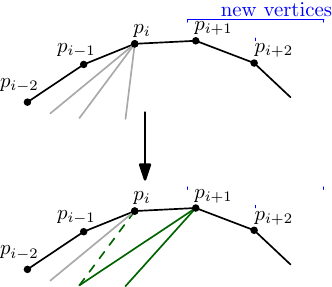}
\caption{Part of a convex chain.
Left: Vertex $p_{i+1}$ has degree $k$ (in this example, $k=3$) and the graph is obtained from one where $p_i$ has degree $k-1$.
This requires the presence of edge $p_{i-1} p_{i+1}$; $p_{i+1}$ inherits all edges incident to $p_{i}$.
Right: Vertex $p_{i+1}$ has degree $k$ (with $k=2$ here) and the graph is obtained from one where $p_{i}$ has degree at least $k$.
In this case, $p_{i+1}$ inherits $k$ edges, and the last inherited edge may be duplicated and remain incident to $p_i$.}
\label{fig:convex_chain}
\end{figure}

Let us now translate this relation to a production matrix.
Suppose we are given the vector $\ourvec{v}^{(i)}$ that contains the number of graphs partitioned by their degree at vertex $p_i$.
For example, we can start with $\ourvec{v}^{(1)} = (1, 0, 0, \dots)^T$.
Now we want to obtain the vector $\ourvec{v}^{(i+1)} = C \ourvec{v}^{(i)}$ by finding an appropriate matrix~$C$.
The $j$th row of $C$ is thus used to produce the number of graphs  where  $p_{i+1}$ has degree $j-1$.
Next we derive the shape of the different rows of $C$.

\paragraph{First row.}
The number of plane graphs where $p_{i+1}$ has degree $0$ is equal to the number of graphs counted in $\ourvec{v}^{(i+1)}$. This gives a first row of $1$s in the matrix~$C$.

\paragraph{Second row.}
If $p_{i+1}$ has degree 1, there are two ways in which $p_{i+1}$ could have been added to the graph.
If the degree of $p_{i}$ is 0, we can add the edge $p_{i-1} p_{i+1}$, and we get a one in the first column of the second row.
Otherwise, $p_{i}$ has degree at least 1, and $p_{i+1}$ can inherit one edge from  $p_{i}$.
Moreover, there is the option of keeping (a copy of) the inherited edge incident to $p_{i}$ without creating any crossing.
In total, for each graph in which $p_{i}$ has degree at least one, that gives two ways for making $p_{i+2}$ have degree 1.
Thus, the rest of the row is made of~$2$s.

\paragraph{Other rows.}
The following rows are analogous, shifted by one column every time:
There are two ways for $p_{i+1}$ to have degree $k$.
Either $k$ edges are inherited from $p_{i+1}$, for which the minimum degree for $p_{i}$ needs to be $k$;
since we can always choose to keep the last inherited edge incident to $p_{i+1}$, we get 2 options every time (cf.~\fig{fig:convex_chain}~(left)).
Otherwise, $p_i$ needs to have exactly $k-1$ edges, which are inherited by $p_{i+1}$;
by adding the edge $p_{i-1} p_{i+1}$, the degree of $p_{i+1}$ becomes~$k$  (see \fig{fig:convex_chain}~(right)).
This results in matrix~$C$ in (\ref{eq:matrix_convex}).

For $n$ vertices, we need that the size of $C$ is at least $n$, and then we can obtain $\ourvec{v}^{(n)} = C^{n-3} \ourvec{v}^{(3)}$.
%
Once the production matrix $C$ is found, the final step is to apply the Perron-Frobenius theorem to obtain an asymptotic lower bound.
Suppose matrix $C$ has some fixed size $m$.
Computing the vector $C^n \ourvec{v}^{(c)}$ gives a vector $\ourvec{v}^{(c+n)}$ with the number of graphs produced after $n$ iterations (i.e., graphs on $n+c$ vertices), with the additional restriction that no graphs with a root vertex of degree at least $m-1$ are produced in this process (which becomes relevant as soon as $m \geq n$).
Nevertheless, the vector $\ourvec{v}^{(n+c)}$ gives the number of a subset of all crossing-free graphs on $n$ vertices (we basically look at a sub-tree of the generating tree).
Therefore, the first entry of $C^n$ gives a lower bound on the total number of graphs.

Now, by the Perron-Frobenius theorem  (see, e.g.,~\cite{meyer}), this number is in $\Omega(r^n)$, where $r$ is the largest eigenvalue of~$C$.
In Section~\ref{sec:bounds_eigenvalue}, we will elaborate in more detail on the actual number of graphs produced by a matrix of fixed size, as well as on the conditions on a matrix to apply the Perron-Frobenius theorem.
In any case, the bottom line is that a lower bound on the number of non-crossing graphs on $n$ vertices can be obtained by choosing $C$ to be reasonably large and computing its largest eigenvalue, either exactly or numerically using mathematical software.

%

\section{Generalized double zig-zag chains}
In this section, we describe the classes of point sets that we will investigate and provide a more detailed outline of our general counting approach.

\begin{figure}
\centering
\includegraphics{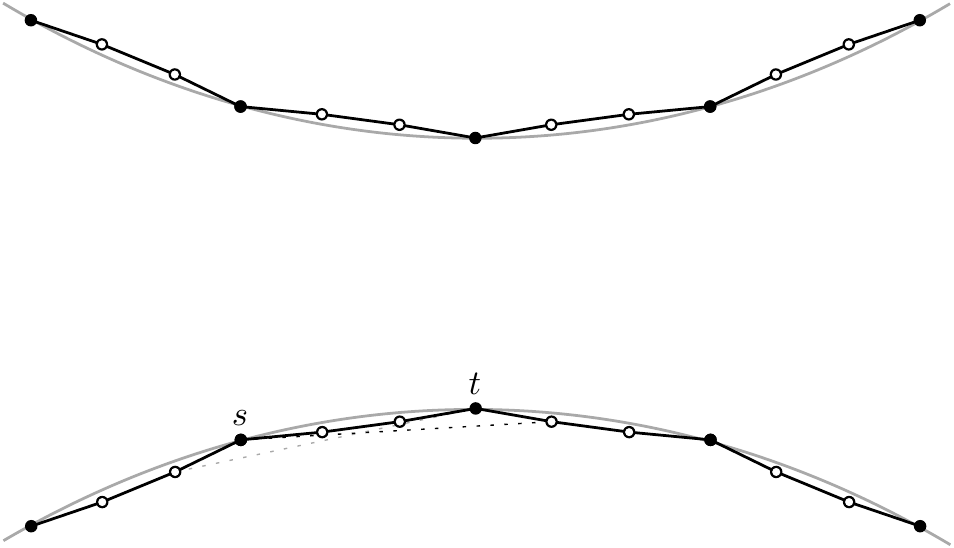}
\caption{A generalized double-zig-zag chain $Z_2$.
The arcs for the construction are gray, the solid edges are not crossed by any segment between two points of the set.
The points placed on the arcs are  black, while the $k=2$ inner points of the pockets are white. 
Any two consecutive black point, such as $s$ and $t$, together with the inner points between them, form one pocket.
The points of the pocket between $s$ and $t$ are above the dotted lines;
This implies that there is no edge between two points of the set that crosses the edge between two consecutive points.}
\label{fig:gdzzc}
\end{figure}

\subsection{The generalized double-zig-zag chain}\label{sec:def_chain}
The construction that we will analyze, and that will allow us to improve on the existing lower bound on the number of crossing-free graphs, is the \emph{generalized double-zig-zag chain}, illustrated in \figurename~\ref{fig:gdzzc}.
It is a family of point sets parameterized by two values $n$ and $k$, where $n$ is the total number of points, and $k$ defines the size of certain re-occurring parts.
We make this more precise in the following paragraph.
(For our construction, we will sloppily refer to that class like to a single point set, as the number of graphs is the same for all sets of that class.)

Let $Z_k$ be a set of $n= 2z$ points with $z \equiv 1 \pmod {(k+1)}$ that is arranged in the following way.
Consider two $x$-monotone circular arcs facing each other as in Fig. \ref{fig:gdzzc}, such that each point on one arc can \emph{see} each point on the other arc (where two points can see each other if the interior of the line segment connecting them does not intersect one of the arcs).
On each arc, we place $\ceil{z/(k+1)}$ points (shown black in the figure).
Consider the segment between two consecutive such points $s$ and $t$ on the lower arc.
We now place a ``flat'' circular arc between $s$ and $t$ with circle center above the arc, and place $k$ points on it (shown white in the figure);
here, flat means that the $k$ points are above any line through $s$ or $t$ and any other point of the construction that is below the line $st$.
We call the group formed by $s$, $t$, and the $k$ points in-between them a \emph{pocket} (with $k$ inner points).
We place $k$ such points between each pair of consecutive points of the lower arc (obtaining the \emph{lower chain}), and also in an analogous way on the upper arc (resulting in the \emph{upper chain}).
The example in \figurename~\ref{fig:gdzzc} shows $Z_2$, where each pocket has two inner points (i.e, $k=2$). 

We label the points along the lower arc, including pockets, from left to right, $p_1, \dots, p_z$, and those on the upper arc $q_1, \dots, q_z$.

\emph{Double chains} (i.e., $k=0$) have provided several of the current best bounds for various settings in the investigation of geometric graphs~\cite{alfredo_lower_bounds,anna}.
In particular, the number of crossing-free graphs on a double chain is in~$\Omega^*((20+14\sqrt{2})^n)$~\cite[Section~5]{alfredo_lower_bounds}.%
\footnote{As customary for exponential bounds, we use the extension of the Landau notation by $O^*(f(n))$, $\Theta^*(f(n))$, and $\Omega^*(f(n))$, in which polynomial factors are omitted; e.g., we have $2^nn^2 \in \Theta^*(2^n)$.}
(This bound was shown to be tight in~\cite{number_plane_geometric}.)
The generalization to the double zig-zag chain (i.e., $k=1$) was devised in~\cite{plane_geometric_soda,number_plane_geometric} to obtain improved bounds on the number of crossing-free graphs and triangulations.
(For $k=1$, each of the chains can be considered as a \emph{double circle}, a point configuration that is conjectured to have the fewest number of triangulations~\cite{bound_double_circle}.)
Generalizations to larger pocket sizes allowed for improving the bound for triangulations~\cite{bounds_multiplicity} and perfect matchings~\cite{asinowsi_matchings}.

\subsection{Counting strategy}
First we observe that the segment between any two consecutive points $p_i p_{i+1}$ is not crossed by any other segment between two points of the set, and thus can co-exist with any other edge in a crossing-free graph.
This also holds for the two edges on the convex hull boundary between the two chains.
For this reason, these edges will be disregarded first in our counting, and will be considered in the end by multiplying the number of graphs not having any such edge by a factor of $2^n$;
this is the number of ways to add these edges.

Therefore, in the next section we will split the counting into two parts.
On the one hand, we will count the graphs with edges below the path $(p_1, \dots, p_z)$ (and, symmetrically, those above the path $(q_1, \dots, q_z)$) called the \emph{outer part}.
On the other hand, we will count the edges that connect vertices of the two paths, which are in the \emph{inner part}.

Our counting will be on~$Z_k$ for $2 \leq k \leq 6$.

\section{Counting for the outer and inner parts}
\label{sec:counting}
In order to count the number of crossing-free graphs in the outer and inner parts of~$Z_k$, we will derive production matrices for them.
We begin with the outer part, for which we will present a matrix that counts the exact number of graphs.
For the inner part, we derive our values by estimating polynomial coefficients, similar to~\cite{number_plane_geometric}.
There are $n$ edges that separate the two parts (connecting consecutive vertices of the chains and the two edges on the convex hull boundary between the chain).
As the number of possibilities to add (or not) such edges to a plane graph is~$2^n$, we will not consider these edges here, and will multiply the resulting bound by $2^n$ in the end.

\subsection{Outer part}
In this section we deduce matrices to count the number of plane graphs with edges below the path $(p_1, \dots, p_z)$, as in \figurename~\ref{fig:almost_convex_chain}.
Recall that a chain is composed of a series of pockets; each pocket with $k$ inner points forms a chain on $k+2$ vertices.
During the explanation, we will consider this path defining a polygon;
the first and the last vertex of a pocket are thus convex vertices of the polygon, and the other vertices are reflex.
The first (say, with smallest index) reflex vertex is called the \emph{leading} vertex of the chain.
We call all other vertices \emph{regular}.

We will present a matrix to count the number of plane graphs in the outer part after adding one whole pocket.
This matrix will be the product of several matrices, one related to each new vertex of the pocket.
For instance, for $k=2$, we will have one related to each of $p_{i+1},p_{i+2},p_{i+3}$ (where $p_i$ is the last vertex of the previous pocket, or, equivalently, the first vertex of the current pocket).

\begin{figure}[tb]
\centering
\includegraphics[page=1]{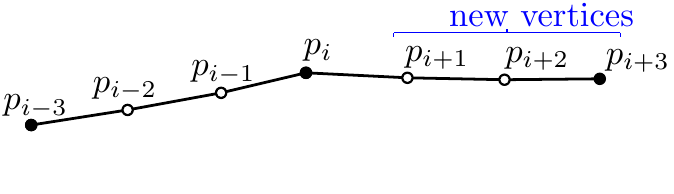}\\
\includegraphics[page=7]{almost_convex_chain2}
\caption{Top: Part of an almost convex chain with two inner vertices (i.e., $k=2$).
Vertices $p_{i-2}$ and $p_{i+1}$ are leading vertices. The other vertices are regular.
Middle: Since $p_{i+2}$ is a regular vertex, any edge incident to $p_{i+2}$ present in a plane graph can be obtained by inheriting an edge from the previous vertex $p_{i+1}$.
The example shows $p_{i+2}$ inheriting two edges from $p_{i+1}$ (bottom).
The last inherited edge may also be kept at $p_{i+1}$ without influencing the degree of~$p_{i+2}$.}
\label{fig:almost_convex_chain}
\end{figure}


\subsubsection{Matrix for regular vertices}
For simplicity, we present the following for $k=2$, but it works in the same way for larger sizes.
Consider a regular vertex like $p_{i+2}$ (refer to \figurename~\ref{fig:almost_convex_chain}).
Our goal is to find a matrix $R$ such that $\ourvec{v}^{(i+2)} = R \ourvec{v}^{(i+1)}$.
\paragraph{First row.}
The plane graphs where $p_{i+2}$ has degree 0 are equal to all the graphs counted in $\ourvec{v}^{(i+1)}$. This gives a first row of $1$s in the matrix $R$.

\paragraph{Second row.}
If $p_{i+2}$ has degree 1, it needs to inherit one edge from $p_{i+1}$.
If the degree of $p_{i+1}$ is 0, this is not possible, thus we get a zero in the first column of the second row.
As soon as $p_{i+1}$ has degree at least 1,  $p_{i+2}$ can inherit one edge from  $p_{i+1}$.
Moreover, there is the option of keeping (a copy of) the inherited edge incident to $p_{i+1}$ without creating any crossing.
In total, for each graph in which $p_{i+1}$ has degree at least one, that gives two ways for making $p_{i+2}$ have degree 1.
Thus  the rest of the row is made of $2$s.

\paragraph{Other rows.}
The following rows are analogous, shifted by one column every time: in order for $p_{i+2}$ to have degree $k$, $k$ edges need to be inherited from $p_{i+1}$, thus the minimum degree for $p_{i+1}$ is $k$.
Since we can always choose to keep the last inherited edge incident to $p_{i+1}$, we get 2 options every time.

This results in matrix~$R$, which is given here for $n=6$:
\begin{equation}\label{eq:R_outer}
\displaystyle R=\left(\begin{array}{rrrrrrrr}
1 & 1 & 1 & 1 & 1 & 1\\
0 & 2 & 2 & 2 & 2 & 2\\
0 & 0 & 2 & 2 & 2 & 2\\
0 & 0 & 0 & 2 & 2 & 2\\
0 & 0 & 0 & 0 & 2 & 2\\
0 & 0 & 0 & 0 & 0 & 2
\end{array}\right) \enspace .
\end{equation}

Exactly the same matrix applies to $p_{i+3}$, and to all other regular vertices when $k>2$.


\subsubsection{Matrix for leading vertices}

Leading vertices like $p_{i+1}$ in \figurename~\ref{fig:almost_convex_chain} require a different approach, as there are edges incident to $p_{i+1}$ that cannot be obtained by inheriting from $p_i$;
that is, edges $p_{i+1}p_{i-1}$, $p_{i+1}p_{i-2}$, $p_{i+1}p_{i-3}$, as $p_{i}p_{i-1}$, $p_{i}p_{i-2}$, $p_{i}p_{i-3}$ are not in the outer part.
(Recall, however, that any edge from $p_i$ to the vertices $p_1, \dots, p_{i-4}$ is completely contained in the outer part.)

In general, for pockets with $k$ inner points,
we partition the graphs depending on which edges connect $p_{i+1}$ to $p_{i-k-1}, \dots, p_{i-1}$.

\begin{figure}[tb]
\centering
\includegraphics[page=8]{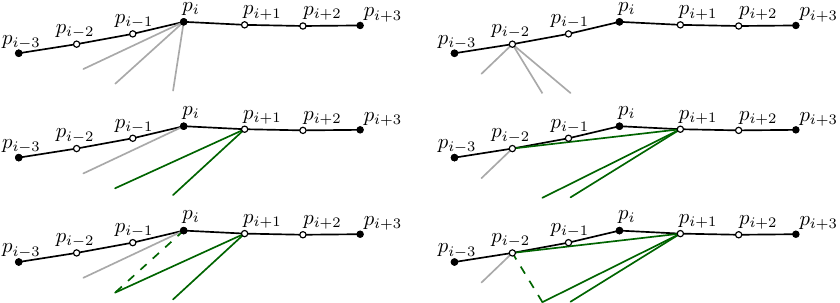}
\caption{
Computing leading vertices on $Z_2$.
Left: When edges $p_{i+1}p_{i-1}$, $p_{i+1}p_{i-2}$ and $p_{i+1}p_{i-3}$ are not included, $p_{i+1}$ can inherit edges from $p_i$. The example shows $p_{i+1}$ inheriting two edges from $p_i$. The last inherited edge (dashed) may be kept without influencing the degree of $p_{i+1}$.
Right: We distinguish cases on which of the edges $p_{i+1}p_{i-1}$, $p_{i+1}p_{i-2}$ or $p_{i+1}p_{i-3}$ are included.
In the example $p_{i+1}p_{i-2}$  is included, and $p_{i+1}$ inherits two edges from $p_{i-2}$.
The dashed edge can be optionally kept. Note that in this case, $p_{i+1}$ cannot inherit any edge from $p_i$.}
\label{fig:almost_convex_chain_case1_2}
\end{figure}

\begin{itemize}
\item \textbf{No edges from $p_{i+1}$ to any of $p_{i-k-1}, \dots, p_{i-1}$.}
All these graphs can be produced by inheriting edges from $p_i$ like for regular vertices, i.e., by applying matrix~$R$; see \fig{fig:almost_convex_chain_case1_2}, left.
\item \textbf{First warm-up: One edge from $p_{i+1}$ to $p_{i-1}$ (but none to $p_{i-k-1}, \dots, p_{i-2}$).}
Such graphs can be produced by inheriting edges from $p_{i-1}$.
The vector for $p_{i-1}$ can be obtained by applying $R^{-1}$ to the one for $p_i$.
(Note that $R$ is invertible, as it is a triangular matrix with non-zero diagonal entries.)
However, inheriting cannot be done in the same way as from $p_i$, as the edge $p_{i+1} p_{i-1}$ increases the degree of $p_{i+1}$ by one.
This increase can be captured by shifting the entries of matrix $R$ vertically one row (i.e., we obtain no graphs with degree $0$, as they have been counted in the previous case, we get exactly one graph of degree one for every parent graph, etc.).
This shift of $R$ is obtained by multiplying it with the matrix $S$ (again given for $n=6$):
\begin{equation}\label{eq:S_outer}
\displaystyle S=\left(\begin{array}{rrrrrrrr}
0 & 0 & 0 & 0 & 0 & 0\\
1 & 0 & 0 & 0 & 0 & 0\\
0 & 1 & 0 & 0 & 0 & 0\\
0 & 0 & 1 & 0 & 0 & 0\\
0 & 0 & 0 & 1 & 0 & 0\\
0 & 0 & 0 & 0 & 1 & 0
\end{array}\right)\enspace .
\end{equation}
The number of graphs with exactly the edge $p_{i+1} p_{i-1}$ added is thus obtained by multiplying with $SRR^{-1} = S$.%
\footnote{Observe that inheriting from $p_i$ and adding the edge $p_{i-1} p_{i+1}$ are exactly the ``operations'' that we used for points in convex position in Section~\ref{sec:tutorial}.
Indeed, $C = R + S$.
While this line of arguments (using the inverse of $R$) is slightly more involved, it generalizes nicely, as we will see in the next two items.}
\item \textbf{Second warm-up: An edge from $p_{i+1}$ to $p_{i-2}$.}
We can apply the same reasoning as before, inheriting from $p_{i-2}$.
The corresponding vector is obtained by multiplying with $R^{-2}$.
We have the edge $p_{i-2} p_{i+1}$, and count both graphs that do and do not contain the edge $p_{i-1} p_{i+1}$.
For the graphs not containing this edge, we apply the shift matrix $S$ once (see \fig{fig:almost_convex_chain_case1_2}, left), and for the ones containing it, we have to apply it twice (as the degree of $p_{i+1}$ is increased by two).
The graphs in this case are thus obtained by multiplying the vector with $SRR^{-2} + S^2RR^{-2} = (S+S^2)R^{-1}$.
\item \textbf{In general, an edge from $p_{i+1}$ to $p_{i-m}$.}
The reader may by now already have realized the pattern to follow for counting graphs with an edge $p_{i+1} p_{i-m}$.
We obtain the vector at $p_{i-m}$ by applying $R^{-m}$.
Then, edges are inherited by applying $R$, but we have to shift $R$ once to account for the degree increase by the edge $p_{i-1}p_{i+1}$ at $p_{i+1}$.
Then we have to consider the different possibilities for edges between $p_{i+1}$ and $p_{i-m+1}, p_{i-m+2}, \dots$;
in general, when adding $l \leq m$ edges between $p_{i+1}$ and those vertices, we have to add $p_{i-m} p_{i+1}$, and have $m-1$ possibilities for the remaining $l-1$ edges.
We get
\[
\sum_{l=1}^m \binom{m-1}{l-1} S^lR^{1-m}
\]
whenever $p_{i-m}$ is the vertex with the smallest index on that pocket with an edge to~$p_{i+1}$.
\end{itemize}

When summing over all $m$, the sum of the matrices is the matrix giving the number of graphs when adding the leading vertex $p_{i+1}$.
Note that any edge between $p_i$ any of the vertices on that pocket, which also includes $p_{i-(k+1)}$, is not in the outer part, and thus graphs with such an edge must not be considered.

\begin{equation}\label{eq:leading_vertex}
\ourvec{v}^{(i+1)} = \left(R + \sum_{m = 1}^{k+1} \sum_{l=1}^m \binom{m-1}{l-1} S^lR^{1-m}\right)\ourvec{v}^{(i)}
\end{equation}

\subsubsection{Putting things together }
The final production matrix for the outer part is obtained by combining matrices $R$ and $S$ accordingly.
For each of the regular vertices it is enough to multiply the previous vector by $R$.
For the leading vertex we use the expression in (\ref{eq:leading_vertex}).
Thus the final combined matrix for the outer part, for pockets with $k$ inner points, is
\begin{equation}\label{eq:pocket}
P = R^k\left(R + \sum_{m = 1}^{k+1} \sum_{l=1}^m \binom{m-1}{l-1} S^lR^{1-m}\right) \enspace,
\end{equation}
and we have $\ourvec{v}^{(i+k+1)} = P \ourvec{v}^{(i)}$.
We again emphasize that multiplying by this matrix accounts for adding a whole pocket (not only for adding a single vertex, like in our introductory example).

Here, we give two instances of $P$ with $k=2$, one of size six and one of size eight.

\[
P_6 = \left(
\begin{array}{cccccc}
 32 & 6 & 26 & \frac{155}{4} & \frac{207}{4} & \frac{271}{4} \\
 48 & 8 & 40 & 62 & 85 & 114 \\
 20 & 4 & 20 & 35 & 52 & 75 \\
 4 & 3 & 9 & 16 & 27 & 44 \\
 0 & 1 & 4 & 6 & 10 & 21 \\
 0 & 0 & 1 & 2 & 2 & 6 \\
\end{array}
\right)
\]

\[
P_8 = \left(
\begin{array}{cccccccc}
 32 & 6 & 26 & 42 & 62 & \frac{327}{4} & \frac{403}{4} & \frac{491}{4} \\
 48 & 8 & 40 & 68 & 104 & 140 & 175 & 216 \\
 20 & 4 & 20 & 40 & 68 & 97 & 126 & 161 \\
 4 & 3 & 9 & 20 & 40 & 62 & 85 & 114 \\
 0 & 1 & 4 & 9 & 20 & 35 & 52 & 75 \\
 0 & 0 & 1 & 4 & 9 & 16 & 27 & 44 \\
 0 & 0 & 0 & 1 & 4 & 6 & 10 & 21 \\
 0 & 0 & 0 & 0 & 1 & 2 & 2 & 6 \\
\end{array}
\right)
\]

So far, we did not discuss the start vector in detail.
Indeed, it will not influence the asymptotic lower bound on the number of graphs.
Still, it is crucial for a valid reasoning, in particular when considering the use of the inverse matrix in~(\ref{eq:leading_vertex}).
However, its structure is simple.
As there are no edges in the outer part between the first $k+1$ vertices of a chain, we can set $\ourvec{v}^{(k+1)} = (1, 0,\dots,0)^T$, denoting that there is a single graph with no edges.

\subsection{Inner part}\label{sec:entropy}
The number of graphs on the inner part can be bounded by an approach which is a generalization of the one presented in~\cite{plane_geometric_soda} that was also used in~\cite{bounds_multiplicity}.
Consider one of the two chains of $Z_k$, which has $z = \frac{n}{2}$ vertices.
Then the number of pockets on the chain is $\frac{z+1}{k+1}$, where each pocket forms a convex chain of $k+2$ elements.
The main idea of the approach is to count the possibilities of putting edges between vertices of a pocket;
we say that such an edge is \emph{covering} the vertices behind them, meaning that covered vertices cannot see any vertex of the other chain.
For the remaining edges, the set of non-covered vertices behaves like a double chain, for which the number of plane graphs is known.
The number of vertices covered will depend on parameters, $\alpha, \beta, \dots$, which we will then optimize.
Let us first formalize the term of covering vertices of a pocket.
An edge $p_j p_{j+l+1}$ between two vertices of a single pocket \emph{covers} the vertices $p_{j+1}, \dots, p_{j+l}$.
That is, if an edge covers $l$ vertices of a pocket, it forms a convex $(l+2)$-gon with the covered vertices.
See \fig{fig:tapped4cup}.

\begin{figure}
\centering
\includegraphics{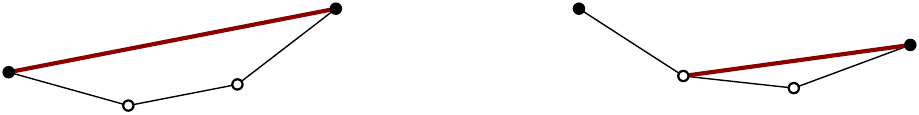}\\
\caption{Left: A pocket with an edge covering all its inner vertices.
The brown thick edge ``hides'' the inner two vertices from the vertices of the other chain.
Right: One vertex of a pocket is covered.}\label{fig:tapped4cup}
\end{figure}

We provide the main idea of our counting strategy for small pocket sizes.
For $Z_2$, the number of pockets of a chain is $\frac{z+1}{3}= \frac{n+2}{6}$.
To construct a plane graph, we first choose $\alpha \frac{z+1}{3} = \alpha \frac{n+2}{6}$ pockets from the chain and cover them.
(In particular, $\alpha$ is a rational number between 0 and 1 such that $\alpha \frac{n+2}{6}$ is an integer;
eventually, we will use $\alpha$ and similar parameters to provide an asymptotic bound, which does not change by rounding as we consider $n$ to be large.)
Then there are three ways to get a plane graph in the convex polygon defined by an edge covering two vertices of a pocket.
This gives us a factor of
\begin{equation}\label{eq:num_in_cups}
{\binom{\frac{n+2}{6}}{\alpha \frac{n+2}{6}}} 3^{\alpha (n+2)/6} \enspace .
\end{equation}

Of the remaining $\frac{n+2}{6} - \alpha \frac{n+2}{6}$ pockets, we choose $\beta \frac{n+2}{6}$ pockets, $0 < \beta < 1$, and cover one point.
(Again, we consider all expressions involving $\alpha$ and $\beta$ to be integers by the choice of the values.)
There are two ways to choose the covered point for each of these pockets, as they are convex chains on four vertices.
This gives us a factor of
\begin{equation}\label{eq:num_covering_one}
{\binom{(1-\alpha) \frac{n+2}{6}}{\beta \frac{n+2}{6}}} 2^{ \beta (n-2)/6}\enspace .
\end{equation}

To get the asymptotic behavior of these factors, we estimate a binomial coefficient using the fact
\[
\frac{2^{H(\lambda)n}}{n+1} \leq {\binom{n}{\lambda n}} \leq {2^{H(\lambda)n}} \enspace ,
\]
where $H(x)$ is the entropy function
\[
H(x)=-x \log_2{(x)} - (1-x)\log_2{(1-x)} \enspace ,
\]
and $\lambda n$ is an integer between 0 and $n$~\cite[Lemma~9.2]{probability_and_computing}.
This implies
\[
{\binom{\mu (n+2)}{\lambda (n+2)}} \in \Theta^*\left (2^{H\left(\frac{\lambda}{\mu}\right) \mu n}\right ) \enspace .
\]
Combining (\ref{eq:num_in_cups}) and (\ref{eq:num_covering_one}) with the entropy function, we get
\begin{equation}\label{eq:num_coverings_approx}
\Theta^*(2^{H(\alpha)n/6} 2^{\log(3)\alpha n/6} \cdot
2^{H(\beta/(1-\alpha))(1-\alpha)n/6} 2^{\beta n/6})
\end{equation}
possibilities for covering the given number of points of one chain in the described way.

Then, the number of points of the chain which are not covered is
\[
z' = \frac{n}{2}- 2\alpha\frac{n+2}{6} - \beta \frac{n+2}{6} = \frac{n}{6}\left(3- 2\alpha - \beta \right) - \frac{2\alpha + \beta}{3} \enspace .
\]

The number of graphs having only edges between the chains is thus the same as for the double chain.
A double chain of $2z'$ vertices has $\Omega^*\left(\left(\frac{10+7\sqrt{2}}{3+2\sqrt{2}}\right)^{2z'}\right)$ such plane graphs~\cite[Section~5]{alfredo_lower_bounds}.
This is obtained by dividing the total number of graphs (see Section~\ref{sec:def_chain}) by the number of graphs on two convex polygons of size~$n/2$ (recall Section~\ref{sec:tutorial}).
By multiplying this twice with (\ref{eq:num_coverings_approx}) (which gives the possibilities for one chain), we get an asymptotic lower bound on the number of plane graphs for our two parameters.
Numerically maximizing over these parameters (yielding $\alpha \approx 0.1396304$ and $\beta \approx 0.3178$) 
gives us a lower bound on the number of plane graphs in the inner part of $Z_2$ of
\[
\Omega(4.18611^n) \enspace .
\]
(Note since the base is rounded, we can again use the classic Landau notation and do not need to consider polynomial factors.)
While it may seem very wasteful to consider only the graphs with a fixed fraction of pockets covered in a certain way, observe that, as $k$ is a constant, the number of different proportions among the covering types is only polynomial in $n$;
thus the dominating factor is the number of graphs for the best choice of $\alpha$ and $\beta$.

In the same way, we obtain bounds for $Z_k$ with larger $k$.
The tedious part of this computation is to obtain the number of possibilities for covering a certain number of vertices.
Note that the number of plane graphs of a pocket with $k$ inner points, without its boundary edges, is given by the formula for the number of plane graphs on $k+2$ points in convex position, divided by $2^{k+2}$.
For example, Figures~\ref{fig:cases_tapped_3} and~\ref{fig:cases_tapped_2} give a full enumeration of the cases for $Z_3$.
The numbers for different pocket sizes are summarized in Table~\ref{tbl:coefficients}.
We give the computations for $Z_5$, as this pocket size gives our best bound.

\begin{lemma}\label{lem:inner_z5}
An instance of $Z_5$ with $n$ points has $\Omega(4.6796443062467462506^n)$ crossing-free geometric graphs with edges only in the inner part.
\end{lemma}
\begin{proof}
For an instance of $Z_5$ with $n$ points, let $\alpha_i$, be the fraction of pockets in which $i$ vertices are covered.
The number of possibilities for the other numbers of covered vertices are given in \fig{fig:cases_tapped_5}.
There, only the maximal covering edges are shown, the remaining ones are indicated by the factors.
For example, the convex region defined by an edge covering two vertices can have three different graphs (as can also be seen in \fig{fig:cases_tapped_2}).
For larger convex regions, we rely on the known numbers from, e.g., \cite[sequence A054726]{oeis}.
For example, for a 7-gon, there are 25216 graphs;
this number has to be divided by $2^7$ as we do not consider edges on the boundary of the 7-gon, giving 197.
For 5 covered vertices (with 197 possibilities in a pocket), we get
\[
{\binom{\frac{n+2}{12}}{\alpha_5 \frac{n+2}{12}}} 197^{\alpha_5 (n+2)/12} \enspace ,
\]
for four covered vertices (and 121 possibilities)
\[
{\binom{(1-\alpha_5)\frac{n+2}{12}}{\alpha_4 \frac{n+2}{12}}} 121^{\alpha_4 (n+2)/12} \enspace ,
\]
for three covered vertices
\[
{\binom{(1-\alpha_5 - \alpha_4)\frac{n+2}{12}}{\alpha_3 \frac{n+2}{12}}} 52^{\alpha_3 (n+2)/12} \enspace ,
\]
etc.
Using the entropy function, the number of ways for covering the given number of points is in
\begin{multline}\label{eq:num_coverings_approx_5}
\Theta^*(2^{\xi(n)n/12}) \enspace \text{with}\\
\begin{split}
\xi(n) = {} &
H(\alpha_5) + \log(197)\alpha_5 + \\
&H(\alpha_4/(1-\alpha_5))(1-\alpha_5) +  \log(121)\alpha_4 + \\
&H(\alpha_3/(1-\alpha_5-\alpha_4))(1-\alpha_5-\alpha_4) + \log(52)\alpha_3) +\\ &H(\alpha_2/(1-\alpha_5-\alpha_4-\alpha_3))(1-\alpha_5-\alpha_4-\alpha_3) + \log(18)\alpha_2) +\\ &H(\alpha_1/(1-\alpha_5-\alpha_4-\alpha_3-\alpha_2))(1-\alpha_5-\alpha_4-\alpha_3-\alpha_2) + \log(5)\alpha_1))
\enspace .
\end{split}
\end{multline}
The number of points of the chain that are not covered is
\[
z' = \frac{n}{2} - \frac{n+2}{12}\sum_{i=1}^5 i \alpha_i \enspace .
\]
Again, we can consider the remaining vertices as being the ones of a double chain and count the graphs in the inner part.
The overall number is the product with~(\ref{eq:num_coverings_approx_5}).
A numerical optimization of the fractions of the pockets with different numbers of covered vertices gives the following values for the parameters.
\begin{equation*}
\begin{split}
\alpha_5\approx {} & 0.0640442057906801992\\
\alpha_4\approx {} & 0.1343042239402862207\\
\alpha_3\approx {} & 0.1970599318079991059\\
\alpha_2\approx {} & 0.2328939317697186669\\
\alpha_1\approx {} & 0.2208748945673803411
\end{split}
\end{equation*}
From those, we obtain the claimed lower bound.
\end{proof}

\begin{table}
\centering
\begin{tabular}{c|c|c|c|c}
$Z_2$ & $Z_3$ & $Z_4$ & $Z_5$ & $Z_6$ \\
\hline
\begin{tabular}[t]{c|c}
t: & p: \\
\hline
1 & 2\\
2 & 3\\
&\\
&\\
&\\
&
\end{tabular} &
\begin{tabular}[t]{c|c}
t: & p: \\
\hline
1 & 3\\
2 & 7\\
3 & 11\\
&\\
&\\
&
\end{tabular} &
\begin{tabular}[t]{c|c}
t: & p: \\
\hline
1 & 4\\
2 & 12\\
3 & 28\\
4 & 45\\
&\\
&
\end{tabular} &
\begin{tabular}[t]{c|c}
t: & p: \\
\hline
1 & 5\\
2 & 18\\
3 & 52\\
4 & 121\\
5 & 197\\
&
\end{tabular} &
\begin{tabular}[t]{c|c}
t: & p: \\
\hline
1 & 6\\
2 & 25\\
3 & 84\\
4 & 237\\
5 & 550\\
6 & 903
\end{tabular}\\
\hline
$4.18^n$ & $4.39^n$ & $4.55^n$ & $4.67^n$ & $4.77^n$\\
\hline
$41.77^n$ & $42.01^n$ & $42.10^n$ & $42.11^n$ & $42.09^n$
\end{tabular}
\caption{Number of possibilities $p$ to cover $t$ vertices in different generalized double zig-zag chains.
The last-but-one line contains the bound for the graphs in the inner part (numbers rounded down).
The bottom line contains the obtained overall bounds.}
\label{tbl:coefficients}
\end{table}

\begin{figure}
\centering
\includegraphics[page=1]{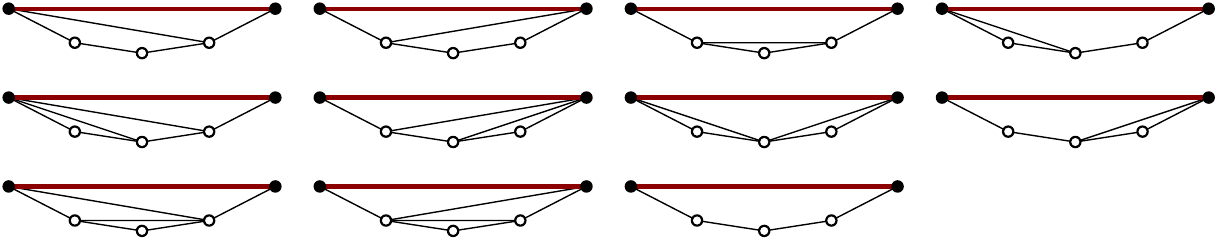}
\caption{Eleven different graphs for a pocket with three inner points in which all three points are covered.}
\label{fig:cases_tapped_3}
\end{figure}

\begin{figure}
\centering
\includegraphics[page=2]{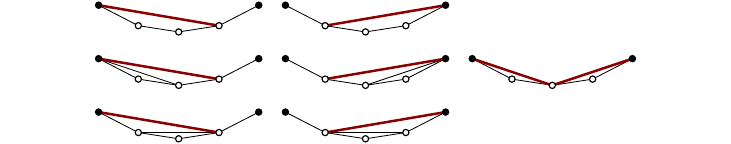}
\caption{Seven different graphs for a pocket with five inner points and two covered points.}
\label{fig:cases_tapped_2}
\end{figure}

\begin{figure}
\centering
\includegraphics{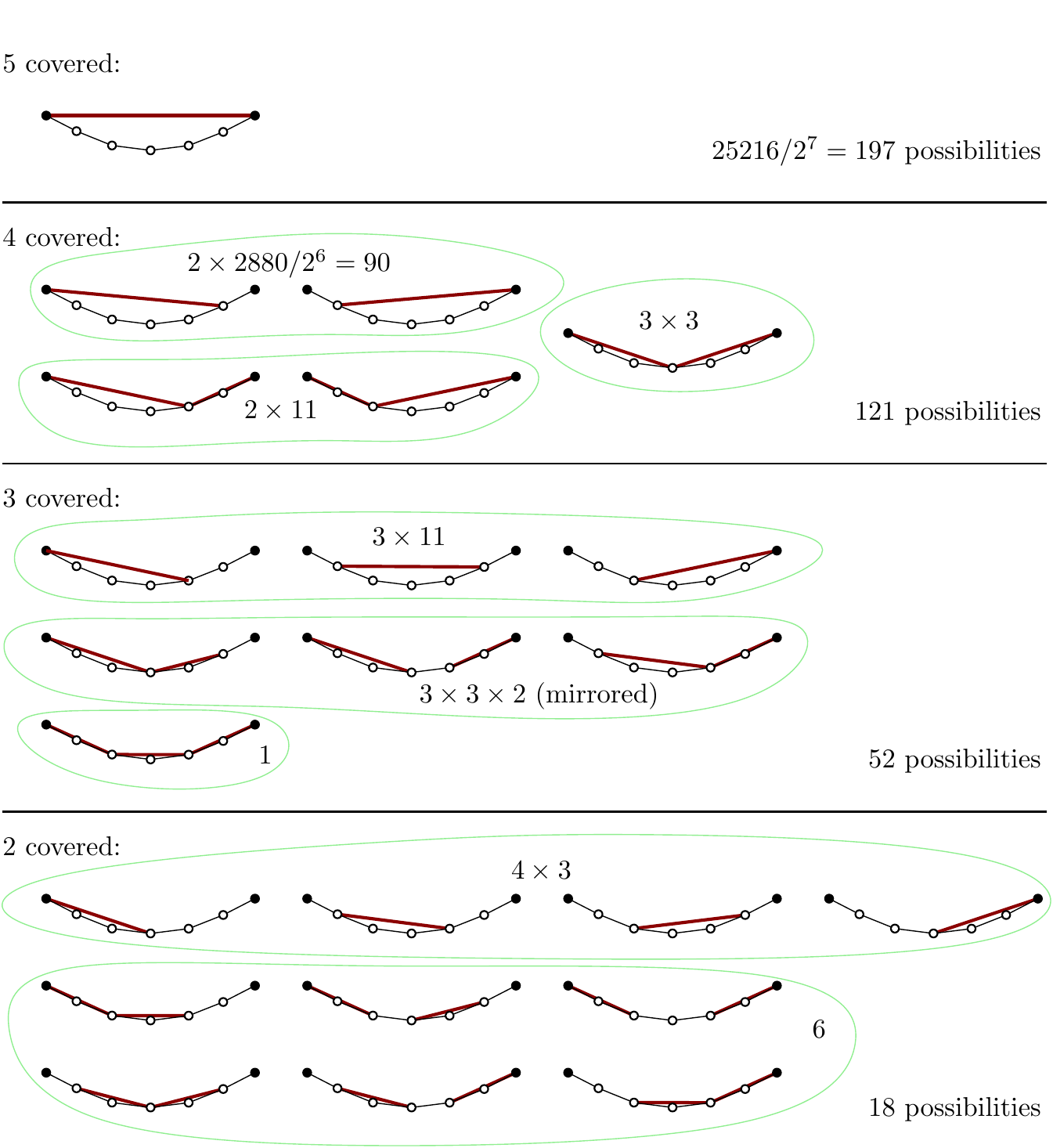}
\caption{Different possibilities for covering pockets with 5 inner points, i.e., convex chains with seven vertices.
We distinguish the different types of paths that cover the vertices, and multiply it with the number of graphs in the convex sub-parts.
The numbers for the convex sub-parts can be obtained from the known number of crossing-free graphs on small point sets in convex position.
See, e.g.,~\cite[sequence~A054726]{oeis}.
}
\label{fig:cases_tapped_5}
\end{figure}

We obtained the values using \emph{Mathematica 11.3}.
For the maximizations of the inner part (i.e., obtaining $\alpha_i$), numerical maximization was used;
hence, we here are not guaranteed to obtain the optimal value, but a lower bound.
Using these values, the computation of the base of the exponent was done with an accuracy of 20 digits.
The values given here are rounded down and thus provide an accurate lower bound, under the assumption of a correctly implemented mathematical software (which, for this computation, means approximating logarithms and roots up to the required accuracy).

\section{A lower bound using the largest eigenvalue}\label{sec:bounds_eigenvalue}
In order to use the production matrices devised to obtain bounds on the number of crossing-free graphs, we need to bound the elements of the matrix powers as $n$ tends to infinity.
This asymptotic information is given by the largest eigenvalue of the production matrix, which is what we analyze next.


A matrix $A$ is \emph{primitive} if it is non-negative (i.e., all its entries are $\geq 0$), and $A^N$ is positive for some natural number~$N$~\cite[p.~678]{meyer}.
Let $A$ be a production matrix of fixed size $m \times m$ that is primitive.
A well-known consequence of the Perron-Frobenius theorem is that
\begin{equation}\label{eq:perron_frobenius}
\lim_{n\to \infty} {\left(\frac{A}{r}\right)}^n = \frac{\ourvec{p}\ourvec{q}^T}{\ourvec{q}^T\ourvec{p}} > 0 \enspace ,
\end{equation}
where $r$ is the Perron root (i.e., largest eigenvalue) of~$A$~\cite[p.~674]{meyer}, and $\ourvec{p}$ and $\ourvec{q}$ are the associated eigenvectors of $A$ and $A^T$, respectively.
Since $\ourvec{p}$ and $\ourvec{q}$ are constant vectors, each entry of $A^n$ is in $\Theta(r^n)$.
An asymptotic lower bound on the entries of the $i$th power of a production matrix provides a means of obtaining the asymptotic number of elements constructed by that production matrix:
multiplying the initial vector $\ourvec{v}^{(n_0)}$ with $A^i$ gives the degree vector for $ci < m$ points.
We use (\ref{eq:perron_frobenius}) to obtain such a lower bound.

However, there is one caveat. 
The previous result is for a matrix $A$ of fixed size $m \times m$, while we would like to consider $n \times n$ matrices, for $n$ approaching infinity.
Still, we observe in the following lemma that this is enough for obtaining lower bounds.
Namely, the $n$th power of a fixed-size $(m \times m)$ production matrix, for some constant $m$, gives a lower bound on the number of graphs on $n$ vertices.


\begin{lemma}\label{lem:constant_matrix_size}
Let $P_m$ be the production matrix for crossing-free geometric graphs on a generalized double-zig-zag chain $Z_k$, as in (\ref{eq:pocket}) of size~$m$, and let $P_n$ be the analogous matrix of size~$n > m$.
Let $C_m(N) = (1, \dots, 1) P_m^N (1, 0, \dots, 0)^T$ and $C_n(N) = (1, \dots, 1) P_n^N (1, 0, \dots, 0)^T$ be the number of graphs on $k(N+1)$ vertices counted by $P_m$ and $P_n$, respectively.
If $m<n$, $C_m(N) \leq C_n(N)$ for any natural number $N$.
\end{lemma}


\begin{proof}
The lemma is a direct consequence of the geometric interpretation of the powers of production matrices.
Observe that we can alter our production rules to produce only graphs with root vertices of degree at most $m-1$.
Clearly, the graphs that we obtain form a subset of all the graphs on the given number of vertices.
(Note, however, that this set of graphs is \emph{not} the one of graphs with degree at most $m-1$.)
In each iteration, we add a vertex, and produce a certain set of graphs from a previously given one.
Each produced graph again has a unique parent, and therefore the number of graphs produced is a lower bound on the total number of crossing-free graphs.
\end{proof}

Lemma~\ref{lem:constant_matrix_size} is particularly remarkable as the interested reader may have noticed that the sub-matrix defined by the first $m$ entries of $P_n$ is different from $P_m$, and it even contains fractions. 
Intuitively, these effects are due to the use of the inverse matrix $R^{-1}$ in (\ref{eq:pocket});
the fact that no graphs with a high-degree root vertex are produced is taken into account by the inverse matrix.

The Perron-Frobenius theorem requires a matrix that is primitive.
However, it turns out that this is not the case for $P$, e.g., when $k=3$;
there, $P$ can already have negative entries.
In the following, we therefore define a related matrix $P'$, to which we will be able to apply the Perron-Frobenius theorem.
Recall the definition of $P$ in~(\ref{eq:pocket}):
\[
P = R^k\left(R + \sum_{m = 1}^{k+1} \sum_{l=1}^m \binom{m-1}{l-1} S^lR^{1-m}\right) \enspace.
\]
We define
\[
P' = \left(R + \sum_{m = 1}^{k+1} \sum_{l=1}^m \binom{m-1}{l-1} S^lR^{1-m}\right)R^k = \left(R^{k+1} + \sum_{m = 1}^{k+1} \sum_{l=1}^m \binom{m-1}{l-1} S^lR^{k+1-m}\right) \enspace .
\]
Observing that $P^j = R^k P'^{j-1} \left(R^{k+1} + \sum_{m = 1}^{k+1} \sum_{l=1}^m \binom{m-1}{l-1} S^lR^{k+1-m}\right)$, we can rewrite
\[
\ourvec{v}^{(n)} = P^{(n-1)/(k+1)-1} \ourvec{v^{(k+1)}}
\]
to
\[
\ourvec{v}^{(n)} = R^k P'^{(n+1)/(k+1)-2} \left(R + \sum_{m = 1}^{k+1} \sum_{l=1}^m \binom{m-1}{l-1} S^lR^{1-m}\right) \ourvec{v}^{(k+1)} \enspace .
\]
Observe that $P'$ models producing the graphs where the root vertex is a leading vertex of a pocket.
We start at $\ourvec{v}^{(k+1)}$, obtain the numbers of graphs at the leading vertex with index $k+2$ (using the matrix in the brackets), and then apply $P'$ for each new pocket;
Finally, we apply $R^k$ to a account for the regular vertices in the last pocket.
To obtain an asymptotic lower bound for the number of graphs, we are thus interested in the entries of powers of $P'$.

\begin{lemma}\label{lem:matrix_ok}
The matrix $P'$ of size~$m$  defined by (\ref{eq:pocket}) is primitive for all~$k>0$.
\end{lemma}
\begin{proof}
$P'$ is the sum of positive powers of $S$ and $R$.
Hence, it is non-negative (and actually, an integer matrix).
It remains to show that there exists a natural number $N$ such that all entries of $P'^N$ are positive.
All entries above the first subdiagonal of $R$ are positive, and hence the same holds for $R^{k+1}$, and thus for $P'$.
For $m = l = 1$, we get the summand $SR^k$.
This matrix has positive entries in the first subdiagonal.
Hence, all entries of $P'$ above or on the first subdiagonal are positive.
Thus, for $N\geq m-1$, the entries of $P'^N$ are positive.
\end{proof}


It follows that we can apply the lemmas above to the production matrix $P'$.
In particular, the largest eigenvalue of each matrix $P'$ (for any combination of $k$ and $m$) will give a lower bound on the number of plane graphs.

We computed the eigenvalues using \emph{Mathematica~11.3}, again with an accuracy of at least 20 digits.
For $Z_2$,
the largest eigenvalue is at least $124.22239555$, when taking the constant-size production matrix of size 1024.%
\footnote{After the presentation of this work at the \emph{European Workshop on Computational Geometry 2018}, G\"unter Rote (personal communication)
applied an extension of a method that was first
used in~\cite[Theorem~12]{asinowsi_matchings}
 for
a geometric counting problem similar to ours:
non-crossing perfect matchings in
repetitively structured point sets.
He derived a
polynomial system for characterizing the largest eigenvalue of our production matrix for $Z_2$.
The numerical solution of the system gives $x \approx
124.225396744416
$.
By trying to find a
polynomial that fits this value, $x$ has experimentally been
established to be
a root of the polynomial $x^3 - 125x^2 + 96x + 28$.}
Recall that multiplying with matrix $P'$ once accounts for adding $k+1=3$ vertices.
For one chain, we thus have $\Omega(\sqrt[3]{124.22239555}^{n/2})$ graphs;
for both outer parts combined, we thus obtain $\Omega(\sqrt[3]{124.22239555}^{n})$.
For the inner part, we obtained $\Omega(4.1861094216284688831^n)$ in Section~\ref{sec:entropy}.
Accounting for the $2^n$ ways to add edges along the chains, we get
$\Omega((\sqrt[3]{124.22239555} \cdot 4.1861094216284688831 \cdot 2)^n) \in \Omega(41.773981586^n)$ crossing-free graphs.

As already mentioned, after studying several values of $k$, the best bound that we were able to find was the one for $Z_5$ (with $m=1024$).
For the inner part, we obtained $\Omega(4.6796443062467462506^n)$ in Lemma~\ref{lem:inner_z5}.
The largest eigenvalue for the matrix obtained from (\ref{eq:pocket}),  with $m=1024$, is $8303.6171640967198428$.
As we add six vertices for each pocket, we need to take the sixth root of this value.
This results in $\Omega((\sqrt[6]{8303.6171640967198428} \cdot 4.6796443062467462506 \cdot 2)^n) \in \Omega(42.116673256039055102^n)$ graphs,
and thus gives the main result of this paper.

\begin{theorem}
There exist sets of $n$ points with $\Omega(42.116673256039055102^n)$ crossing-free graphs.
\end{theorem}

\section{A note on mixing pocket sizes}
The construction used in the previous section uses pockets with a fixed number $k$ of inner points.
However, one could consider using  more than one pocket size.
Namely, consider a point set formed by two chains like a generalized double-zig-zag chain, but where the pockets have different numbers of points.

One interesting observation is that the order of the pockets does not matter.
For the inner part, this follows from the counting in Section~\ref{sec:entropy};
the number only depends on the number and variants of coverings.
For the outer part, we can use a similar argument.
The two outer parts of such chains correspond to so-called \emph{almost convex polygons} (in which the points are connected by the polygon boundary from left to right), that were previously considered by Hurtado and Noy~\cite{almost_convex}.
They made a statement for triangulations that is analogous to the following.

\begin{proposition}
Consider an almost convex polygon $P$ with two adjacent pockets $A = (p_1, \dots, p_k)$ and $B=(p_k, \dots p_l)$ with convex vertices $p_1$, $p_k$ and $p_l$.
Then the almost convex polygon $P'$ in which the pockets $A$ and $B$ are swapped has the same number of plane graphs as~$P$.
\end{proposition}
\begin{proof}
If the two pockets have the same number of points (i.e., $l = 2k - 1$), then there is nothing to prove.
We map the set of plane graphs on $P$ to the set of plane graphs on $P'$, as before disregarding edges on the boundary.

Let $p'_1, \dots p'_l$ be the points on the boundary of $P'$, and note that $p_{l-k+1}$ is a convex vertex of $P'$.
If a graph does not contain an edge between the two pockets, then we can map each edge $ab$ in $P$ to the corresponding edge $a'b'$ in $P'$, which gives a bijection between these graphs.
For a plane graph $G$ with an edge between these two pockets, let $p_i p_j$ be the edge such that $i$ is minimal and $j$ is maximal (i.e., the edge covers all other such edges from the ``interior'').
We map this graph to a plane graph $G'$ on $P'$ in the following way.
The edge $p_i p_j$ is mapped to the edge $p'_{l-i+1} p'_{l-j+1}$.
In both polygons, we have now a chain of $l-j+i+1$ vertices;
for edges with only one endpoint on the pockets, we map that endpoint to the corresponding endpoint of the chain.
All other edges $ab$ are mapped to $a'b'$.
The region bounded by the pockets and $p_i p_j$ is the mirror image of the region bounded by the pockets and $p'_{l-i+1} p'_{l-j+1}$ in $P'$, which also defines a mapping for the edges inside these regions.
Hence, the number of plane graphs in $P$ and $P'$ is the same.
\end{proof}

Therefore, to count the number of plane graphs with different pocket sizes, we merely have to multiply a constant number of matrices for different $k$, and we get the matrix for a longer chain that is a combination of pockets of different sizes.
However, after experimenting with several such combinations of pocket sizes, we did not obtain improved bounds.

\section{Conclusions}
We slightly improved the previously best lower bound on the maximum number of crossing-free geometric graphs on $n$ points using production matrices.
Applying production matrices to families of well-structured point sets appears to be a conceptually simple way of obtaining bounds for important families of graphs.

It is interesting that with this technique it is also possible to obtain bounds when using a mix of different pocket sizes.
While we could not find combinations that improve the presented bound in this way, our search was not exhaustive, and we cannot rule out that such an approach could allow to improve the lower bound even further.


\paragraph*{Acknowledgments.}
We thank  G\"unter Rote and Andr\'e Schulz for valuable discussions.

\bibliographystyle{abbrv}

\begin{thebibliography}{1}

\bibitem{lower_bound_triangulations}
O.~Aichholzer, V.~Alvarez, T.~Hackl, A.~Pilz, B.~Speckmann, and B.~Vogtenhuber.
\newblock An improved lower bound on the minimum number of triangulations.
\newblock In S.~P. Fekete and A.~Lubiw, editors, {\em 32nd International
  Symposium on Computational Geometry (SoCG 2016)}, volume~51 of {\em LIPIcs},
  pages 7:1--7:16. Schloss Dagstuhl - Leibniz-Zentrum fuer Informatik, 2016.

\bibitem{convexity_pseudo_triangulations}
O.~Aichholzer, F.~Aurenhammer, H.~Krasser, and B.~Speckmann.
\newblock Convexity minimizes pseudo-triangulations.
\newblock {\em Comput. Geom.}, 28(1):3--10, 2004.

\bibitem{number_plane_geometric}
O.~Aichholzer, T.~Hackl, C.~Huemer, F.~Hurtado, H.~Krasser, and B.~Vogtenhuber.
\newblock On the number of plane geometric graphs.
\newblock {\em Graphs Combin.}, 23:67--84, 2007.

\bibitem{plane_geometric_soda}
O.~Aichholzer, T.~Hackl, B.~Vogtenhuber, C.~Huemer, F.~Hurtado, and H.~Krasser.
\newblock On the number of plane graphs.
\newblock In {\em Proc. 17th Annual {ACM-SIAM} Symposium on Discrete Algorithms
  ({SODA} 2006)}, pages 504--513. {ACM} Press, 2006.

\bibitem{bound_double_circle}
O.~Aichholzer, F.~Hurtado, and M.~Noy.
\newblock A lower bound on the number of triangulations of planar point sets.
\newblock {\em Comput. Geom.}, 29(2):135--145, 2004.

\bibitem{ajtai}
M.~Ajtai, V.~Chv{\'a}tal, M.~Newborn, and E.~Szemer{\'e}di.
\newblock Crossing-free subgraphs.
\newblock In {\em Theory and Practice of Combinatorics}, pages 9--12.
  North-Holland, 1982.

\bibitem{asinowsi_matchings}
A.~Asinowski and G.~Rote.
\newblock Point sets with many non-crossing perfect matchings.
\newblock {\em Comput. Geom.}, 68:7--33, 2018.

\bibitem{eco_survey}
E.~Barcucci, A.~D. Lungo, E.~Pergola, and R.~Pinzani.
\newblock {ECO}: a methodology for the enumeration of combinatorial objects.
\newblock {\em J. Differ. Equations Appl.}, 5(4-5):435--490, 1999.

\bibitem{deutsch}
E.~Deutsch, L.~Ferrari, and S.~Rinaldi.
\newblock Production matrices.
\newblock {\em Adv. in Appl. Math.}, 34(1):101--122, 2005.

\bibitem{bounds_multiplicity}
A.~Dumitrescu, A.~Schulz, A.~Sheffer, and {\relax Cs}.~D. T{\'{o}}th.
\newblock Bounds on the maximum multiplicity of some common geometric graphs.
\newblock {\em {SIAM} J. Discrete Math.}, 27(2):802--826, 2013.

\bibitem{guillermo}
G.~Esteban~Pascual.
\newblock Production matrices and enumeration of geometric graphs, 2018.
\newblock Master's thesis, {Universitat Polit\`ecnica de Catalunya}.

\bibitem{FN}
P.~Flajolet and M.~Noy.
\newblock Analytic combinatorics of non-crossing configurations.
\newblock {\em Discrete Mathematics}, 204(1-3):203--229, 1999.

\bibitem{alfredo_lower_bounds}
A.~Garc{\'{\i}}a~Olaverri, M.~Noy, and J.~Tejel.
\newblock Lower bounds on the number of crossing-free subgraphs of {$K_N$}.
\newblock {\em Comput. Geom.}, 16(4):211--221, 2000.

\bibitem{treeOfTrees}
M.~C. Hernando, F.~Hurtado, A.~M{\'{a}}rquez, M.~Mora, and M.~Noy.
\newblock Geometric tree graphs of points in convex position.
\newblock {\em Discrete Appl. Math.}, 93(1):51--66, 1999.

\bibitem{anna}
C.~Huemer and A.~de~Mier.
\newblock Lower bounds on the maximum number of non-crossing acyclic graphs.
\newblock {\em Eur. J. Comb.}, 48:48--62, 2015.

\bibitem{production_matrices_geometric_graphs}
C.~Huemer, A.~Pilz, C.~Seara, and R.~I. Silveira.
\newblock Production matrices for geometric graphs.
\newblock {\em Electr. Notes Discrete Math.}, 54:301--306, 2016.

\bibitem{characteristic_polynomials}
C.~Huemer, A.~Pilz, C.~Seara, and R.~I. Silveira.
\newblock Characteristic polynomials of production matrices for geometric
  graphs.
\newblock {\em Electr. Notes Discrete Math.}, 61:631--637, 2017.

\bibitem{almost_convex}
F.~Hurtado and M.~Noy.
\newblock Counting triangulations of almost-convex polygons.
\newblock {\em Ars Comb.}, 45:169--179, 1997.

\bibitem{treeOfTriangulations}
F.~Hurtado and M.~Noy.
\newblock Graph of triangulations of a convex polygon and tree of
  triangulations.
\newblock {\em Comput. Geom.}, 13(3):179--188, 1999.

\bibitem{MeVe}
D.~Merlini and M.~C. Verri.
\newblock Generating trees and proper {R}iordan arrays.
\newblock {\em Discrete Math.}, 218(1--3):167--183, 2000.

\bibitem{meyer}
C.~D. Meyer.
\newblock {\em Matrix analysis and applied linear algebra}.
\newblock SIAM, Philadelphia, 2000.

\bibitem{probability_and_computing}
M.~Mitzenmacher and E.~Upfal.
\newblock {\em Probability and computing - randomized algorithms and
  probabilistic analysis}.
\newblock Cambridge University Press, 2005.

\bibitem{sharir_sheffer_charging}
M.~Sharir and A.~Sheffer.
\newblock Counting plane graphs: Cross-graph charging schemes.
\newblock {\em Combin. Probab. Comput.}, 22(6):935--954, 2013.

\bibitem{sheffer_webpage}
A.~Sheffer.
\newblock Some plane truths.
\newblock https://adamsheffer.wordpress.com/numbers-of-plane-graphs/.
\newblock Retrieved June 14, 2018.

\bibitem{oeis}
N.~J.~A. Sloane.
\newblock The on-line encyclopedia of integer sequences.
\newblock https://oeis.org.
\newblock Retrieved June 14, 2018.

\end{thebibliography}

\end{document}